\documentclass[12pt]{article}

\usepackage{graphicx, amssymb, latexsym, amsfonts, amsmath, lscape, amscd,
amsthm, color, epsfig, mathrsfs, tikz, enumerate}
\usepackage{soul}

\newtheorem{theorem}{Theorem}[section]

\newtheorem{lemma}[theorem]{Lemma}

\newcommand\DELETE[1]{}

\begin{document}

\title{{\bf On local structures of cubicity 2 graphs}}
\author{
{\sc Sujoy Kumar Bhore$^{(a)}$}, {\sc Dibyayan Chakraborty$^{(b)}$}, {\sc Sandip Das$^{(b)}$}, \\{\sc Sagnik Sen$^{(c)}$}\\
\mbox{}\\
{\small $(a)$ Ben-Gurion University, Beer-Sheva, Israel}\\
{\small $(b)$ Indian Statistical Institute, Kolkata, India}\\
{\small $(c)$ Indian Statistical Institute, Bangalore, India}
}

\date{}

\maketitle

\maketitle
\begin{abstract}
 A 2-stab unit interval graph (2SUIG)  
is an axes-parallel unit square intersection graph where
the unit squares  intersect either of the two fixed lines parallel to the $X$-axis,
distance $1 + \epsilon$ ($0 < \epsilon < 1$) apart. This family of graphs allow us to study local structures of  unit square intersection graphs, that is, graphs with cubicity 2.  
The complexity of determining whether a tree has cubicity 2 is unknown while the graph recognition problem 
for unit square intersection graph is known to be NP-hard.  
We present a polynomial time algorithm for recognizing  trees that admit a 2SUIG representation.
 \end{abstract}


\section{Introduction}
Cubicity $cub(G)$ of a graph $G$ is the minimum  $d$ such that $G$ is representable as a geometric intersection graph of 
$d$-dimensional (axes-parallel) cubes~\cite{roberts}. The notion of cubicity is a special case of boxicity~\cite{roberts}. 
Boxicity $box(G)$ of a graph $G$ is the minimum  $d$ such that $G$ is representable as a geometric intersection graph of 
$d$-dimensional  (axes-parallel) hyper-rectangles. 
Given a graph $G$ it is NP-hard to decide if $box(G) \leq n$~\cite{krato} and  $cub(G) \leq n$~\cite{breu} for all  $n \geq 2$. 
On the other hand, the family of graphs with boxicity 1 and the family of graphs with cubicity 1 are just the families of interval and unit interval graphs, respectively. The graph recognition  problems  for both these families are solvable in polynomial time.

Trees with boxicity 1 are the caterpiller graphs while all trees have boxicity 2. 
 On the contrary, determining cubicity of a tree  seems to be a more difficult problem. It is easy to note that trees with cubicity 1 are paths. 
For higher dimensions, Babu et al.~\cite{Chandran2014} presented a randomized algorithm that 
runs in polynomial time and computes cube representations of trees, 
of dimension within a constant factor of the optimum. The complexity of determining the cubicity of a tree is unknown~\cite{Chandran2014}.

In a recent work~\cite{Bhore2015}, some new families of graphs, 
based on the local structure of boxicity 2 and cubicity 2 graphs were introduced and studied. 
 A \textit{2-stab unit interval graph (2SUIG)}  
is an axes-parallel unit square intersection graph where
the unit squares  intersect either of the two fixed lines, called upper and lower stab lines, parallel to the $X$-axis,
distance $1 + \epsilon$ ($0 < \epsilon < 1$) apart (see Fig.~\ref{2SUIG_example} for example). 
For convenience, 
let $y = 1$ be the \textit{lower stab line}  and 
$y = 2+ \epsilon $ be the \textit{upper stab line} where $ \epsilon \in (0,1)$  is a constant, for the rest of the article. 
 The family of such graphs are called the \textit{2SUIG} family, introduced~\cite{Bhore2015} for studying the local structures of cubicity 2 graphs.
 A geometric representation of the above mentioned type of a graph  is called a \textit{2SUIG representation} of the graph. 
 Given a 2SUIG representation of a graph $G$, the vertices corresponding to the unit squares intersecting the upper stab line are called \textit{upper} vertices and the vertices corresponding to the unit squares intersecting the lower stab line are called \textit{lower} vertices.
 If a set of vertices are all lower (or upper) vertices then we say they are in the \textit{same stab}.

In this article, we characterize all   trees that admit a 2SUIG representation using  forbidden structures. In particular,  we prove the following:

\begin{theorem}\label{main_th_tree}
Determining whether a given tree $T = (V, E)$ is a 2SUIG can be done in $O(|V|)$ time. 
\end{theorem}

 If a tree is a 2SUIG, then
our algorithm can be used to find a 2SUIG representation of it. Our algorithm finds a forbidden structure responsible for the tree
not having a 2SUIG representation.  
In particular, 2SUIG is a graph family with cubicity (and boxicity) 2 which contains the family of unit interval graphs as a subfamily. 
Moreover, the family of 2SUIG graphs is not perfect~\cite{golumbic} as 5-cycle has a 2SUIG representation (see Fig.~\ref{2SUIG_example}). 
So our work, to the best of our knowledge,  is
the first non-trivial work on recognizing subclass of trees with cubicity  2.

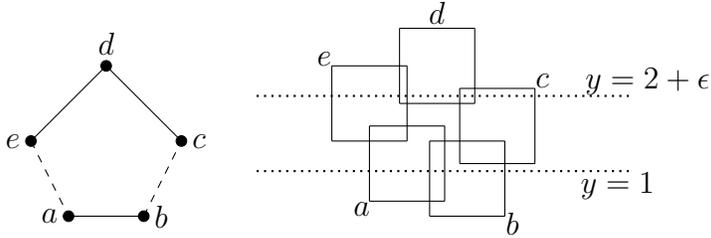
\begin{figure}

\centering
\begin{tikzpicture}

\filldraw [black] (0,0) circle (2pt) {node[left]{$a$}};
\filldraw [black] (1,0) circle (2pt) {node[right]{$b$}};
\filldraw [black] (1.5,1) circle (2pt) {node[right]{$c$}};
\filldraw [black] (-.5,1) circle (2pt) {node[left]{$e$}};
\filldraw [black] (.5,2) circle (2pt) {node[above]{$d$}};

\draw[-] (0,0) -- (1,0);
\draw[-] (-.5,1) -- (.5,2);
\draw[-] (1.5,1) -- (.5,2);

\draw[dashed] (-.5,1) -- (0,0);
\draw[dashed] (1,0) -- (1.5,1);


\draw[-] (4,.2) -- (5,.2);
\draw[-] (4,1.2) -- (5,1.2);
\draw[-] (4,.2) -- (4,1.2);
\draw[-] (5,.2) -- (5,1.2);

\node at (3.9,.1) {$a$};

\draw[-] (4.8,0) -- (5.8,0);
\draw[-] (4.8,1) -- (5.8,1);
\draw[-] (4.8,0) -- (4.8,1);
\draw[-] (5.8,0) -- (5.8,1);

\node at (5.9,-.1) {$b$};

\draw[-] (3.5,1) -- (4.5,1);
\draw[-] (3.5,2) -- (4.5,2);
\draw[-] (3.5,1) -- (3.5,2);
\draw[-] (4.5,1) -- (4.5,2);

\node at (3.4,2.1) {$e$};

\draw[-] (4.4,1.5) -- (5.4,1.5);
\draw[-] (4.4,2.5) -- (5.4,2.5);
\draw[-] (4.4,1.5) -- (4.4,2.5);
\draw[-] (5.4,1.5) -- (5.4,2.5);

\node at (4.9,2.7) {$d$};

\draw[-] (5.2,.7) -- (6.2,.7);
\draw[-] (5.2,1.7) -- (6.2,1.7);
\draw[-] (5.2,.7) -- (5.2,1.7);
\draw[-] (6.2,.7) -- (6.2,1.7);

\node at (6.3,1.8) {$c$};


\draw[thick,dotted] (2.5,1.6) -- (7.5,1.6);

\node at (7.7,1.8) {$y = 2 + \epsilon $};

\draw[thick,dotted] (2.5,.6) -- (7.5,.6);

\node at (7.3,.4) {$y = 1 $};

\end{tikzpicture}

\caption{A representation (right) of a $2SUIG$ graph (left).}\label{2SUIG_example}

\end{figure}

\section{Preliminaries}
To prove our results we will need several standard and non-standard definitions which we will present in this section.

Let $G$ be a graph. The set of vertices and edges are denoted by $V(G)$ and $E(G)$, respectively. A vertex subset $I$ of $G$ is an \textit{independent set} if all the vertices of $I$ are pairwise non-adjacent. The cardinality of the largest independent set of $G$ is its \textit{independence number}, denoted by $\alpha(G)$.

Let $G$ be a unit square intersection  graph with a fixed representation $R$. We denote the unit square in $R$ corresponding to the vertex $v$ of $G$ by  $s_v$. In this article, by a unit square we will always mean a closed unit square. 
The co-ordinates of the left lower-corner of $s_u$ is denoted by $(x_u,y_u)$.
Given a graph $G$ with a 2SUIG representation $R$ and two vertices $u,v \in V(G)$ we say   $s_u <_x s_v$  if $x_u < x_v$ and 
$s_u <_y s_v$  if $y_u < y_v$ (see figure~\ref{fig_order}).
Let $H$ be a connected subgraph of $G$. Consider the union of intervals obtained from the projection of 
unit squares corresponding to the vertices of $H$ on $x-$axis. 
This so obtained interval $span(H)$ is called the \textit{span} of $H$ in $R$.

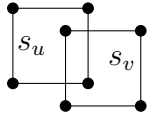
\begin{figure}
\centering
\begin{tikzpicture}

\filldraw [black] (0,0) circle (2pt) {node[left]{}};
\filldraw [black] (1,0) circle (2pt) {node[right]{}};
\filldraw [black] (1,1) circle (2pt) {node[left]{}};
\filldraw [black] (0,1) circle (2pt) {node[right]{}};

\draw[-] (0,0) -- (1,0) -- (1,1) -- (0,1) -- (0,0);

\filldraw [black] (.7,-.3) circle (2pt) {node[left]{}};
\filldraw [black] (1.7,-.3) circle (2pt) {node[right]{}};
\filldraw [black] (1.7,.7) circle (2pt) {node[left]{}};
\filldraw [black] (.7,.7) circle (2pt) {node[right]{}};

\draw[-] (.7,-.3) -- (1.7,-.3) -- (1.7,.7) -- (.7,.7) -- (.7,-.3);

\node at (.25,.5) {$s_u$};

\node at (1.45, .3) {$s_v$};

\end{tikzpicture}
\caption{In the above picture $s_u <_x s_v$ and $s_u <_y s_v$.}
\label{fig_order}
\end{figure}

A \textit{leaf} is a vertex with degree 1. 
A \textit{caterpillar} is a tree where every leaf vertex is adjacent to a vertex of a fixed path.
A \textit{branch vertex} is a vertex having degree more than 2.
A \textit{branch edge} is an edge incident to a branch vertex.
A \textit{claw} is  the 
complete bipartite graph $K_{1,3}$.
Given a 2SUIG representation $R$ of a graph $G$, an edge $uv$ is a \textit{bridge edge} if $s_u$ and $s_v$ intersect different stab lines.

Let $P = v_1v_2...v_k$ be a path with a 2SUIG representation $R$. 
The path $P$ is a \textit{monotone path} if either $s_{v_1} <_x s_{v_2} <_x ... <_x s_{v_k}$ or $s_{v_k} <_x s_{v_{k-1}} <_x ... <_x s_{v_1}$.
Let $P = v_1v_2...v_k$ be a monotone path with $s_{v_1} <_x s_{v_2} <_x ... <_x s_{v_k}$
 and all $v_i$'s are in the same stab. Observe that 
 $ \alpha(P) = \lceil \frac{k}{2} \rceil < span(P) \leq k$ are tight bounds for $span(P)$. 
 Fix some constant $c \in (0, .5)$ for the rest of this article. 
A monotone representation of $P$ is \textit{stretched} if $span(P) = k$ and is \textit{shrinked} if $span(P) = \lceil \frac{k}{2} \rceil + c$ (see Fig.~\ref{fig shrinked-streached}). 
 The value of $c$ will not affect our proof. 
If all the vertices of $P$ are in the same stab, then $P$ must be monotone.
 So, a path like $P$  can have four different such 2SUIG representations:

\begin{figure}
\center
\includegraphics[scale=.5]{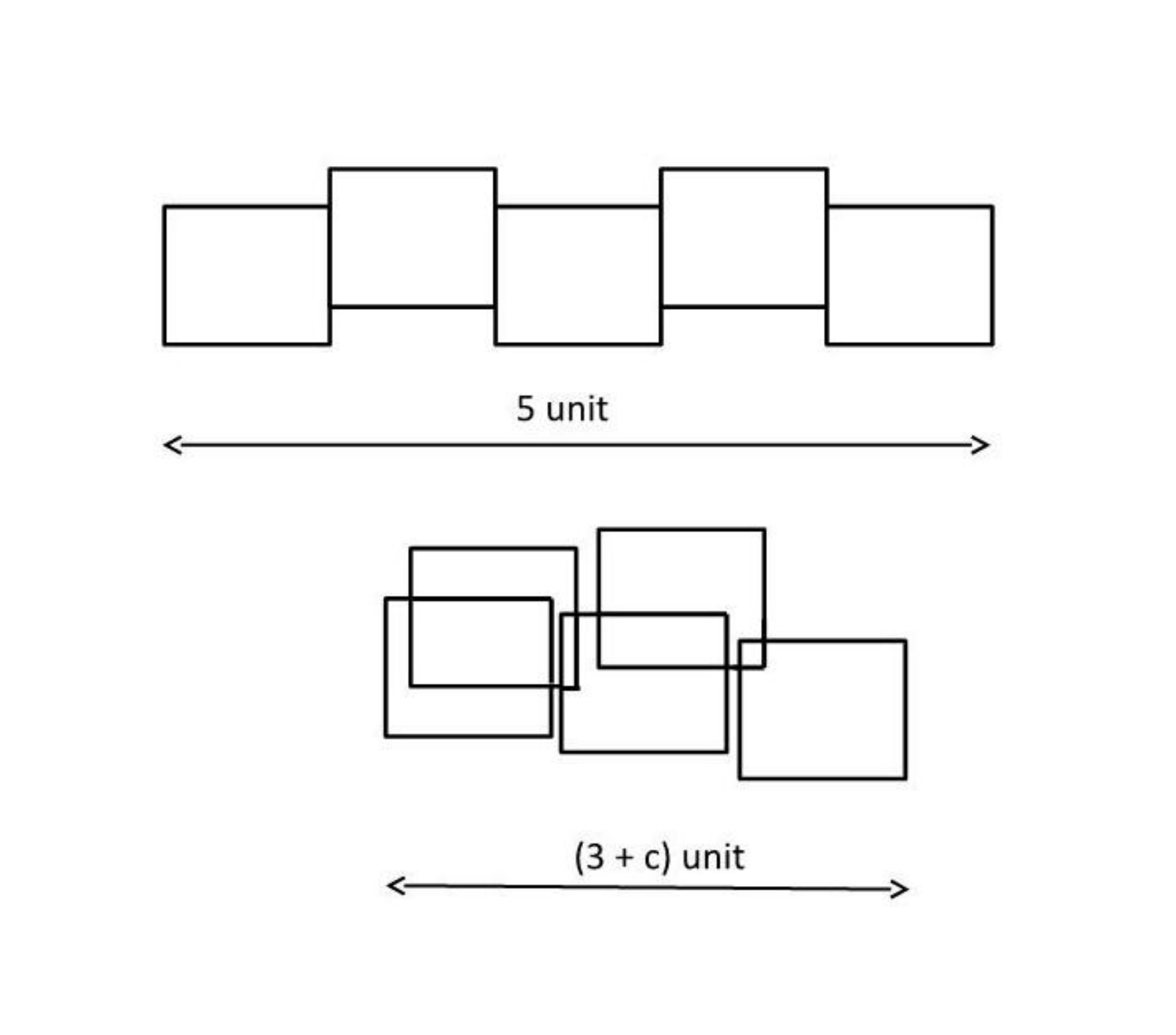}
\caption{A streached (above) and shrinked (below) representation of a path on five vertices.}
\label{fig shrinked-streached}
\end{figure}

 \begin{itemize}
 \item[1.] \textit{lower-right monotone}: all the vertices of the path are lower vertices with $s_{v_1} <_x s_{v_2} <_x ... <_x s_{v_k}$;
 
 \item[2.] \textit{upper-right monotone}: all the vertices of the path are upper vertices with $s_{v_1} <_x s_{v_2} <_x ... <_x s_{v_k}$;
 
 \item[3.]  \textit{lower-left monotone}: all the vertices of the path are lower vertices with $s_{v_k} <_x s_{v_{k-1}} <_x ... <_x s_{v_1}$; 
 
 \item[4.] \textit{upper-left monotone}: all the vertices of the path are upper vertices with $s_{v_k} <_x s_{v_{k-1}} <_x ... <_x s_{v_1}$.
 \end{itemize}

The first two types of monotone paths are called \textit{right monotone} and last two types are called \textit{left monotone}.

A path $P = v_1v_2...v_k$ is called a \textit{folded path} if it has a 
degree two vertex $u$ such that either 
$s_u <_x s_v$ for all $v \in V(P) \setminus \{u\}$  or $s_v <_x s_u$ for all $v \in V(P) \setminus \{u\}$. 

A \textit{red edge} of a tree $T$ is an edge $e$ 
 such that  each  component of $T \setminus \{e\}$ contains a claw.
A \textit{red path} is a path induced by red edges. 
A \textit{maximal red path} is a red path that is not properly contained in another red path.
Let $P = v_1v_2...v_k$ be a  maximal red path  in $T$. The vertices $v_1$ and $v_k$ are \textit{endpoints} of $P$.

\section{Proof of Theorem~\ref{main_th_tree}}
Given a tree $T$, here our objective is to determine whether $T$ has a 2SUIG representation.
Let  $T = (V,E)$ be a fixed tree. 
We will now prove several necessary conditions for $T$ being a 2SUIG. 
On the other hand,  we will show that these conditions together are also sufficient and can be verified in $O(|V|)$ time.

\subsection{Structural properties}\label{subsec structure}
We will start by proving some structural properties of $T$ assuming it is a 2SUIG.

\begin{lemma}\label{obs_deg_4}
If $T$ has a 2SUIG representation, then its vertices have degree at most four.  
\end{lemma}

\begin{proof}
Let $u$ be a vertex of $T$. Any unit square intersecting $s_u$ must contain  one of the four corners of $s_u$. 
Therefore, if $u$ has degree greater than four, then at least two  unit squares will contain the same corner of $s_u$ creating a 
3-cycle. 
\end{proof}

Note that the above lemma  can easily be extended for trees that are unit sqaure intersection graphs. 

\medskip

A tree $T$ that admits a 2SUIG representation does not necessarily have a red edge. 
But if $T$ has at least one red edge then the red edges of $T$ must induce a path.

\begin{lemma}\label{lm_RUIC}
If $T$ has a 2SUIG representation, then either $T$ has no red edge or the set of red edges of $T$  induces  a connected path.
\end{lemma}

\begin{proof}
Let $T$ has at least one red edge and $T'$ be the graph induced by all red edges.
First we will show that $T'$ is connected. Thus, assume that $T'$ has at least two components $T_1$ and $T_2$. 
Then there is  a  path $P$ in $T$ connecting $T_1$ and $T_2$. 
Note that removing an edge $e$ of $P$ creates two components of   $T$ each of which contains a claw. 
 Thus, $e$ should be  a red edge. Therefore, $T'$ is connected.

Now we will show that $T'$ is  a path. 
Assume that $v$ is a vertex of $T'$ with degree at least 3. Also, let $v_1,v_2$ and $v_3$ be three neighbors of $v$ in $T'$. 
In any 2SUIG representation of $T$, at least three corners of $s_v$ must be intersected 
by $s_{v_1}, s_{v_2}$ and $s_{v_3}$. 
Without loss of generality, we assume that 
 $s_{v_1}$ intersects
the upper-left  corner of $s_v$, $s_{v_2}$ intersect the upper-right  corner of $s_v$, and $s_{v_3}$ intersects the left lower-corner of $s_v$.
This implies that $s_{v_1}, s_{v_2}$ intersect the upper stab line while $s_{v_3}$ intersects the lower stab line. 

Note that a claw has a 2SUIG representation. Any such representation of a claw will have squares intersecting the upper stab line and squares intersecting the lower stab line. 
As each component of $T \setminus \{vv_1\}$ has a claw, there must be a path of the form 
$v_1v_{11}v_{12}...v_{1k}$ in $T$ such that $s_{v_{1i}} <_x s_{v_1}$ for all $i \in \{1, 2, ..., k\}$ where $s_{v_{1k}}$  intersects the lower stab line. 
Similarly, as each component of $T \setminus \{vv_2\}$ has a claw, there must be a path of the form 
$v_2v_{21}v_{22}...v_{2k'}$ in $T$ such that $s_{v_2} <_x s_{v_{2i}}$ for all $i \in \{1, 2, ..., k'\}$ where $s_{v_{2k'}}$  intersects the lower stab line.
Moreover, as each component of $T \setminus \{vv_3\}$ has a claw, there must be a path of the form 
$v_3v_{31}v_{32}...v_{3k''}$ in $T$  where $s_{v_{3k''}}$  intersects the upper stab line. 
This will force a cycle in the representation of $T$, a contradiction.
Thus, $T'$ must be a path.  
\end{proof}

 The above result leads us to two cases:  when $T$ has a red path and when $T$ does not have any red path. 
Note that by Lemma~\ref{lm_RUIC}, if $T$ is a 2SUIG, then either $T$ has no red edge or the red edges induces a path  in $T$.  
If the red edges of $T$ induces a path $P$, then construct the \textit{extended red path} $A = a_1a_2...a_k$ by including the edge(s), that are not red,  incident to the endpoint(s) of $P$ that have degree two in $T$.
In particular, if both the end points of $P$ are branch vertices, the extended red path  $A = P$. 
On the other hand, if 
$T$ has no red edges, then  distance between any two branch vertices is at most 2.
Thus, there exist a vertex $v$ in $T$ whose closed neighborhood $N[v]$ contains all the branch vertices of $T$. 
Choose (if not found to be unique) one such special vertex $v$. If $v$ has degree two then consider the path $uvw$ induced by the closed neighborhood of $v$ and call it the extended red path of $T$. If $v$ does not have degree two, then 
the extended red path of $T$ is the singleton vertex $v$. In any case, rename the vertices of the extended red path 
$A = a_1a_2...a_k$  so that we can speak about it in an uniform framework along with the case $T$ having red edges.
We fix such an extended red path $A= a_1a_2...a_k$ for the rest of this article. 
The vertices $V_A = \{a_1, a_2, ..., a_k\}$ of this extended red path $A$ are called the \textit{red vertices}.

\begin{lemma}\label{lm_no_RUIC}
If $T$ has a 2SUIG representation and does not have any red edge, then the number of branch vertex in $T$ is at most 5.
\end{lemma}

The above follows directly from the fact that the vertices of $T$ has degree at most four. 

\begin{lemma}\label{lm_1branch}
If $T$ has  at most one branch vertex with degree at most four, then $T$ has a 2SUIG representation.
\end{lemma}

\begin{proof}
If $T$ has no branch vertices, then $T$ is a path which admits a unit interval representation.
On the other hand, $T$ with one branch vertices of degree at most four is a subdivision of $K_{1,3}$ or $K_{1,4}$.
These graphs clearly admit 2SUIG representation.
\end{proof}

The next result will provide  intuition about how a tree having a 2SUIG representation looks like. 

\begin{lemma}\label{lem all branch}
A branch vertex of a tree $T$ is either a red vertex or is adjacent to a red vertex. 
\end{lemma}

\begin{proof}
If $T$ has no red edges, then there exists a red vertex $v$ in $T$ such that all the branch vertices are in 
the closed neighborhood $N[v]$  of $v$. 

Thus suppose that  the red edges of  $T$ induces a path. 
Let  a red vertex $v$ and  a non-red branch vertex $u$
be connected by a path with no red edges of length  at least 2. Clearly after deleting this path, the component containing 
$v$  contatins a claw. Thus if we delete the edge $\{e\}$ of the path incident to $v$, then  both the components of 
$T \setminus \{e\}$ contains a claw, a contradiction. 
\end{proof}

Note that if $T$ is a 2SUIG tree with  at least two branch vertices, then the endpoints $a_1$ and $a_k$ of the  extended path $A$ must be branch vertices of $T$. 
Assume that $A_b = \{a_1 = a_{i_1}, a_{i_2}, ..., a_{i_{k'}} = a_k \}$ be the branch vertices of $A$ where $1= i_1 < i_2 < ... < i_{k'} = k$.

The neighbors of  red vertices that are not red are called \textit{agents}.  An agent $v$ is adjacent to exactly one red vertex, 
say $a_j$,
of $T$. We call $v$ is an agent of $a_j$ in this case. 

If we delete all the red vertices and agents from a 2SUIG tree $T$ then by Lemma~\ref{lem all branch} we will be left with some disjoint paths. 
Each such path actually starts from (that is, is one endpoint of) one of the agents. Let $P= v_1v_2...v_l$ be a path where $v_1$ is an agent and the other vertices are neither agent nor red vertices. Also  $v_2, v_3, ..., v_{l-1}$ are degree 2 vertices and $v_l$ is a leaf. 
Then the path $P' = v_2v_3...v_{l-1}v_l$  is called a \textit{tail} of  agent $v_1$. Let $v_1$ be an agent of the red vertex $a_j$. 
Sometimes we will also use the term ``tail $P'$ of the red vertex $a_j$''.  
 Deleting the tail $P'$  is to delete all the vertices of $P'$. 
 The red vertex $a_j$,  all its agents and tails are together called \textit{$a_j$ and its associates} for each $j \in \{1,2,...k\}$. 
  Note that an agent  has exactly two tails by allowing tails with zero vertices. 
 Let us set the following conventions: the tails of an agent $z$  are the \textit{long tail} $lt(z)$ 
 and the \textit{short tail} $st(z)$ such that $|lt(z)| \geq |st(z)|$ where
 $|lt(z)|$ and $|st(z)|$ denotes the number of vertices in the respective tails. 
Now we have enough nomenclatures (see Fig.~\ref{fig nomenclature}) to  present the rest of the proof.

\begin{figure}
\center
\includegraphics[scale=.5]{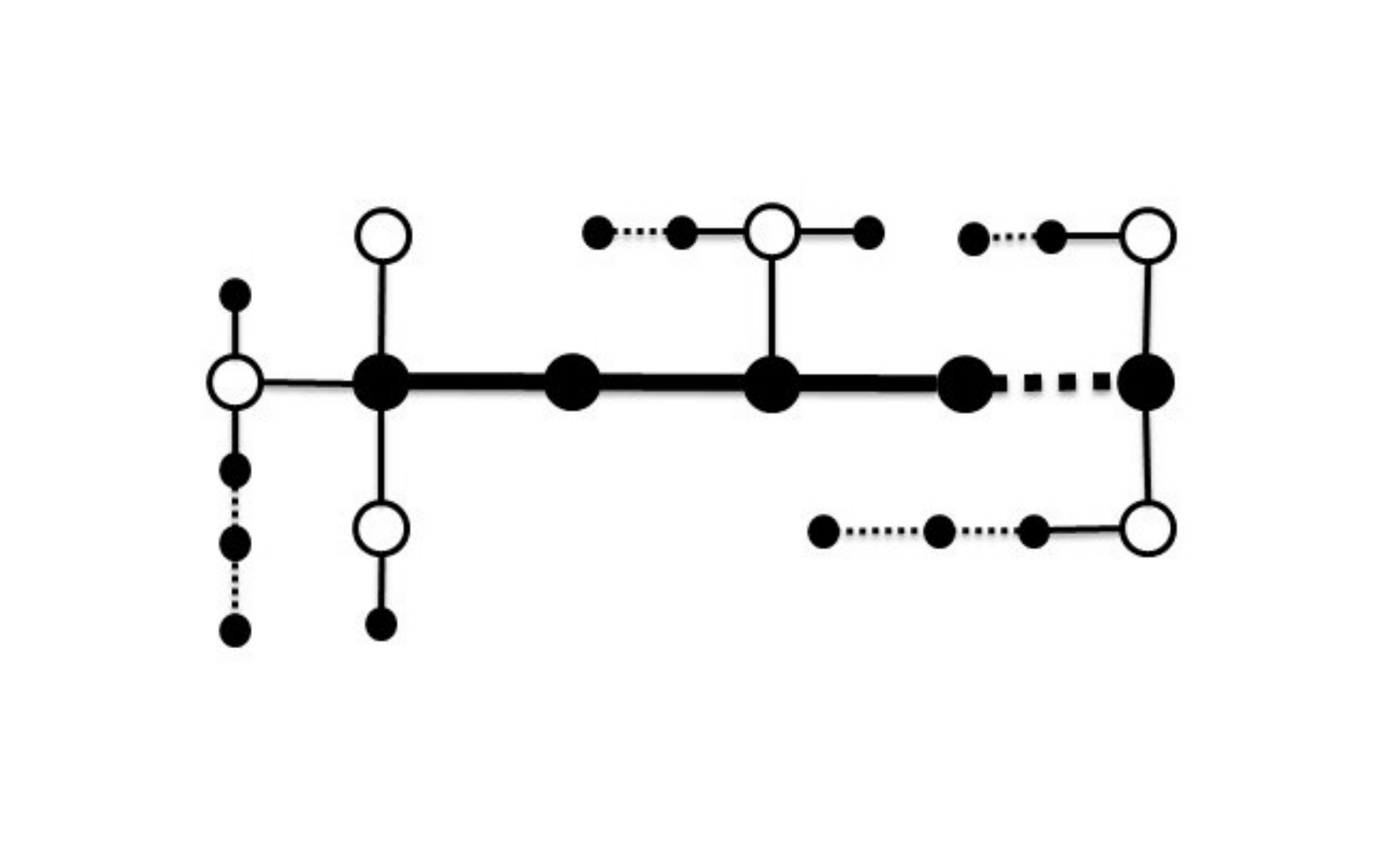}
\caption{The nomenclature -- red edges = thick lines; new edge(s) added to the red path for obtaining the extended red path =  thick dotted line(s);  red vertices = big solid circles, agents = big hollow circles, tails = thin dotted lines; vertices of the tails = small solid circles.}
\label{fig nomenclature}
\end{figure}

Observe that if we delete all the tails of a 2SUIG tree $T$, then we are left with a caterpillar with maximum degree at most 4. 
This is an interesting observation as we know that a tree is an interval graph if and only if it is a caterpillar graph.

\subsection{Partial description of the canonical representation}\label{subsec partial canonical}
Till now we have proved some structural properties of a 2SUIG tree. Now we will discuss about the structure of its 2SUIG representation.  In the following, we will show that there is a canonical way to represent  a 2SUIG tree. 
First we will describe the representation of the extended red path followed by representation of the agents and their tails.

\begin{lemma}\label{lem red monotone}
If $T$ is a 2SUIG with at least one red edge, then there exists  a 2SUIG representation where 
the extended red path of $T$ is monotone. 
\end{lemma}

\begin{proof}
Let $A$ be the extended red path of $T$ with a representation $R$. 
If $A$ is not monotone then one of the following is true. (i) $A$ is a folded path.
(ii) There are three vertices $\{a_{i},a_{i+1},a_{i+2}\}\in V(A)$ such that $s_{a_{i+1}} <_x s_{a_i}$ and $s_{a_{i+1}} <_x s_{a_{i+2}}$ where $i>1$ 
and for all $j<i$ we have $s_{a_j} <_x s_{a_{j+1}}$.

\bigskip

 (i)  There is a vertex $u\in V(A)$ with $s_v <_x s_u$ for all $v \in V(A) \setminus \{u\}$. 
Then there will be two claws $C_1,C_2$ in two different components of $T \setminus \{u\}$ with $s_v <_x s_u$ for all $v \in V(C_1) \cup V(C_2)$. 
But as $T$ is a 2SUIG, this configuration will force a cycle in it. This is a contradiction. 

\medskip

(ii) Without loss of generality assume $s_{a_i}$ intersects the lower stab line.
Then $s_{a_{i+1}}$ and $s_{a_{i+2}}$ must intersect the upper stab line. 
Let $T_{i+1}$ be the component of $T$ obtained by deleting  the edge $a_{i}a_{i+1}$ and contains $a_{i+1}$. 
There is a claw $C_3$ in $T_{i+1}$ with  $s_{a_{i+1}} <_x s_w$ for all $w \in V(C_3)$.
Thus, an agent $z$ of  $a_{i}$ with $s_{a_i} <_x s_z$ is not a branch vertex as otherwise this will force a cycle. 
Hence its tail can be presented by a lower-right monotone representation.
Therefore, we can translate (rigid motion) the component $T_{i+1}$ to the right 
to obtain a 2SUIG representation of $T$ where the extended red path $A$ is monotone.
\end{proof}

Similarly it can be shown that if $T$ is a 2SUIG with no red edge, still it admits a representation 
where 
the extended red path $A$ is monotone. The proof can be argued in a way similar  
to the proof of Lemma~\ref{lem red monotone}. 

Now we will show that it is possible to streach the monotone extended  red path without much problem. The following result is also applicable for trees without red edges. 

\begin{lemma}\label{lem red stretched}
If $T$ admits a 2SUIG representation $R$ with a monotone extended red path, then there exists a 2SUIG representation where $A$ is stretched.
\end{lemma}

\begin{proof}
Let $e = a_ia_{i+1}$ be an edge of the extended red path $A$ with $x_{a_i} < x_{a_{i+1}}$ in $R$. Let $T_i$ and $T_{i+1}$ be the components of $T \setminus \{e\}$ containing $a_i$ and $a_{i+1}$, respectively. Now translate (rigid motion) the component $T_{i+1}$ to the right obtaining a 
2SUIG representation with $x_{a_{i+1}} = x_{a_i} +1$. We are done by performing this operation on every edge of $A$. 
\end{proof}

We turn our focus on the bridge edges of the extended red path.

\begin{lemma}\label{lem red bridge-branch}
If $T$ admits a 2SUIG representation $R$ with a stretched monotone extended red path,  
then there exists a 2SUIG representation where every red bridge vertex is a branch vertex.
\end{lemma}

\begin{proof}
Let $A = a_1a_2...a_k$ be a right monotone extended red path of $T$ with respect to $R$. 
Let $e = a_ia_{i+1}$ be a bridge edge of the extended red path $A$ where $a_{i+1}$ is not a branch vertex. 
Also assume that a bridge vertex $a_j$ is a branch vertex for all $j < i$.

Let $S = \{v \in V(T) | s_{a_i} <_x s_v \}$. Let $T'$ be the graph induced by $S$. Note that $S$  contains the vertex $a_{i+1}$. 
Now consider the reflextion of the 2SUIG representation of $T'$ induced by $R$ with respect to the $X$-axis. 
This will give us a 
picture where every unit square corresponding to the vertices of $T'$ lies under the $X$-axis. This is a particular unit square representation $R'$ of 
$T'$ which in fact is also a 2SUIG representation  of it if we consider the stab lines to be $y= -1$ and $y = -2- \epsilon$. 
In this 2SUIG representation $R'$, the lower vertices with respect to $R$ of $T'$ became upper vertices and vice versa. 
Now translate the unit square representation $R'$ upwards until all the upper vertices (with respect to $R'$) 
of $T'$ intersects $y= 2+ \epsilon$, all the lower vertices intersect $y=1$. 
Note that $a_i$ can have at most one degree 2 agent in $T'$ and thus,   that agent can have at most one tail. 
After what we did above, we can adjust the $Y$-co-ordinates of that agent and its tail, if needed, to obtain a 2SUIG representation of $T$.

We will be done by induction after handling one more case. The case where $a_i$ is the first bridge vertex of $A$ which is not a branch vertex while $a_{i+1}$ is a branch vertex.  Let $S' = \{v \in V(T) | s_v <_x s_{a_{i+1}} \}$ and let $T''$ be the graph induced by $S'$.
 To achieve our goal, we do the exact same thing with $T''$ that we did with $T$. This will provide us a 2SUIG representation of $T$  
 where each bridge vertex $a_j$ is a branch vertex for all $j \leq i+1$.
 Hence we are done by induction.
\end{proof}

\begin{figure}
\center
\includegraphics[scale=.4]{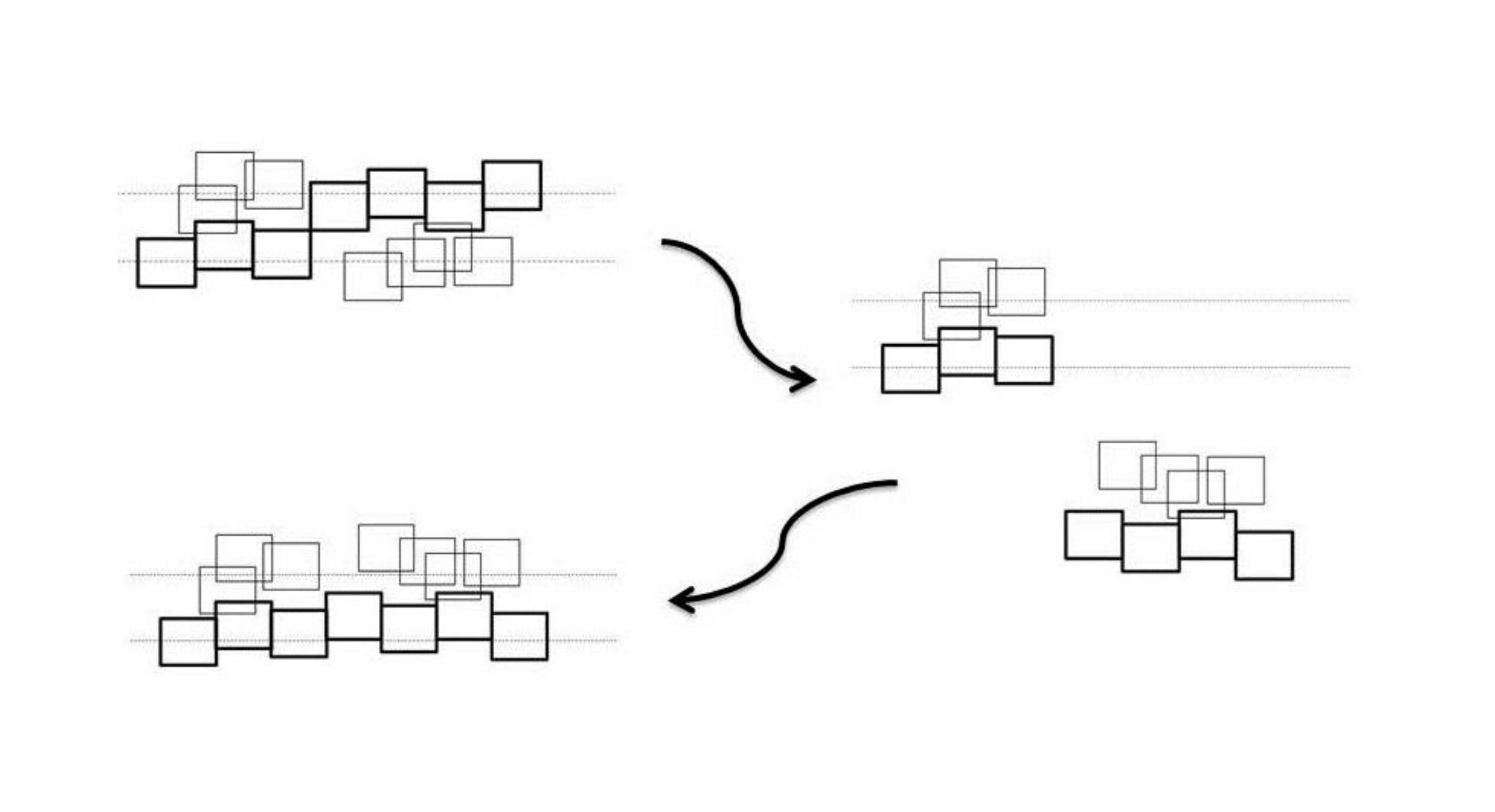}
\caption{Explaining the reflection and translation described in the proof of Lemma~\ref{lem red bridge-branch} with an example.}
\label{fig reflection-translation}
\end{figure}

After proving that  the extended red path is stretched and monotone we will prove the opposite for  tails in the following sense:
  
\begin{lemma}\label{lem tail shrink}
If $T$ admits a 2SUIG representation $R$,  
then there exists a 2SUIG representation where each tail is a shrinked monotone path and all its vertices
are in the same stab. 
\end{lemma}

\begin{proof}
Let $P = v_2 v_3 ...v_{l-1}$ be a tail of agent $v_1$. Note that if all vertices of the tail $P$ are  in the same stab then $P$ must be monotone. 
Furthermore, if $P$ is not shrinked in $R$ then we can shrink it to obtain a new representation of $T$ without changing anything else 
of $R$. 

Therefore, to complete the proof, let us assume that not all vertices of $P$ are in the same stab. Then at least one edge $e$ of $P$ is a bridge edge. Any brige edge divides the stab lines into two parts, left and right. Assume, without loss of generality, that  $s_{v_1}$ is in the left 
part. Thus, as there are no branch vertices in the tail, we do not have any vertex $w \in V(T) \setminus V(P)$ with $s_w$ lying in the right part.
Thus, we can modify the representation $R$ by placing all the vertices of the tail $P$ in the same stab making use of the empty right part.   
\end{proof}

\subsection{Properties of the canonical representation}\label{subsec properties}
According to the preceding discussions, we know that if $T$ is a 2SUIG, then it admits a representation where its 
extended red path is a  monotone stretched path. Without loss of generality, assume that $R$ is such 
a 2SUIG representation of $T$. 
The other vertices of $T$ are the agents and the vertices of the tails. 
   Note that the endpoints $a_1$ and $a_k$ can have at most 3 agents and 6 tails while the other red vertices can have at most 2 agents and 4 tails.

  \begin{lemma}\label{lem one right monotone}
Let   $T$ be a 2SUIG tree with a representation  $R$ where the extended red path $A$ is a stretched  right monotone path and $R$  satisfies the conditions of 
Lemma~\ref{lem red bridge-branch} and \ref{lem tail shrink}. 
\begin{itemize}
\item[$(a)$]  If each right monotone tail of $T$ is such that it is not possible to make the tail left monotone and obtain a 2SUIG 
representation of $T$ from $R$ without changing anything else, then any red vertex other than $a_k$ has at most one right monotone tail 
having at least two vertices.

\item[$(b)$]  If each left monotone tail of $T$ is such that it is not possible to make the tail right monotone and obtain a 2SUIG 
representation of $T$ from $R$ without changing anything else, then any red vertex other than $a_1$ has at most one left monotone tail 
having at least two vertices.
\end{itemize} 
\end{lemma}

\begin{proof}
$(a)$    Let $a_j$ ($j < k$) be a red vertex of $T$ with at least two monotone tails. Without loss of generality, assume that 
   $a_j$ is a lower vertex. Note that $s_{a_{j+1}}$ contains either the upper-right corner or the lower-right corner of $s_{a_j}$.
      Let $v$ be an agent of $a_j$. 
   
   \begin{itemize}
   \item[Case $1.$] If $s_v$ contains the lower-left corner of  $s_{a_j}$, then $v$ cannot have a right monotone tail
   for any $j \geq 2$ as $s_{a_{j-1}}$ contains the upper-left corner of $s_{a_j}$. If  $j=1$, then a right monotone tail 
   $P = v_1v_2...v_l$ of $v$ 
   must be upper-right monotone. This means $vv_1$ is a bridge edge. This will mean, there is no upper vertex $z$ in $T$ 
   with $s_z <_x s_{v_1}$. Thus, if we make the tail $P$ an upper-left monotone tail instead (we keep the position of $s_{v_1}$ as before but change the positions of the other vertices of $P$), then we obtain a 2SUIG representation of $T$. But this should not be possible according to the assumptions. Hence $v$ cannot have a right-monotone tail for any red vertex $a_j$ ($j \neq k$).

   \item[Case $2.$]   If $s_v$ contains the upper-left corner of  $s_{a_j}$, then $v$ can indeed have an upper-right monotone tail. 
   Note that, in this case, $v$ can have at most one right monotone path, an upper-right monotone that is.  
   If some other neighbor of $a_j$ contains the upper-right corner of $s_{a_j}$, then the right monotone tail of $v$ can have at most one vertex in order to avid cycles in $T$.

   \item[Case $3.$]   If $s_v$ contains the upper-right corner of  $s_{a_j}$, then $v$ can have at most one right monotone tail,
   an upper-right monotone tail to be specific,  as this situation implies that $s_{a_{j+1}}$ contains the lower-right corner of $s_{a_j}$. 
   \end{itemize}
   
   Therefore, only the agents containing upper-left corner or upper-right corner of $s_{a_j}$ 
   can have at most one right monotone tail each. But if both types of agents are present, then 
   the right monotone tail of the agent containing upper-left corner of $s_{a_j}$ can have at most one vertex.     
   
\medskip

$(b)$ This proof can be done similarly like $(a)$.    
  \end{proof}
   
   From  the above  we can infer the following:

    \begin{lemma}\label{lem one-one}
Let   $T$ be a 2SUIG tree with a representation  $R$ where 
a lower (or upper) red vertex has two upper (or lower) neighbors. 
\begin{itemize}
\item[$1.$]   If the left neighbor is an agent then it can have at most one right monotone tail having at most one vertex.

\item[$2.$]   If the right neighbor is an agent then it can have at most one left monotone tail having at most one vertex.

\item[$3.$]   If both the neighbors are agents then the left neighbor can have a right monotone tail and the right neighbor 
can have a left monotone tail having one vertex each.
\end{itemize}
\end{lemma}

In the above lemma, the third case is to say that the worst case scenarios of the first two cases can take place simultaniously. 
Now we will discuss the length of the tails that can be accommodated between two red branch vertices.

  \begin{lemma}\label{lem 4options}
Let   $T$ be a 2SUIG tree with a representation  $R$ with a stretched right monotone 
extended red path $A$. Let $a_l$, $a_{l+m}$ ($m=0$ is possible) be two red branch vertices such that 
$a_l$ has a right monotone tail $P$ and $a_{l+m}$ has a left monotone tail $P'$ in the same stab with no  vertex $v$ 
of $T$  satisfying $s_{p} <_x s_v <_x s_{p'}$ where $p, p'$ are the leaf vertices of the tails $P, P'$, respectively.  Then, depending on the positions of the corresponding agents $v$  of $P$ and  $v'$ of $P'$, 
 $R$ must satisfy one of the following conditions:

\begin{itemize}
\item [1.] If $s_{a_l}  <_x s_{v}$ and $s_{v'} <_x s_{a_{l+m}}$, then 
$  \alpha(P_v) + \alpha(P'_{v'})   \leq    m$;
\item [2.] If $s_{v} <_x s_{a_l} $ and $s_{v'} <_x s_{a_{l+m}}$, then
$  \alpha(P_v) + \alpha(P'_{v'}) -1 \leq m$;
\item [3.] If $s_{v} <_x s_{a_l} $ and $s_{a_{l+m}} <_x s_{v'}$, then 
$  \alpha(P_v) + \alpha(P'_{v'}) -2 \leq m$; 
\item [4.] If $s_{a_l}  <_x s_{v}$ and $s_{a_{l+m}} <_x s_{v'}$, then
$  \alpha(P_v) + \alpha(P'_{v'}) -1 \leq m$;
\end{itemize}

\noindent where $P_v$ is the path induced by $V(P) \cup \{v\}$ and $P'_{v'}$ is the path induced by $V(P') \cup \{v'\}$. 
\end{lemma} 

\begin{proof}
Let $A' = a_la_{l+1}...a_{l+m}$ and  $s_{v} <_x s_{a_l} $ and $s_{a_{l+m}} <_x s_{v'}$. As $A'$ is stretched and 
$P_v$ and $P'_{v'}$ are both shrinked we have
$ span(A') + 2 \geq span(P_v) + span(P'_{v'})$. This implies 
$m+1 + 2 = m+3 \geq  \lceil  \alpha(P_v) + \alpha(P'_{v'}) +2c \rceil = \alpha(P_v) + \alpha(P'_{v'}) +1$ and hence condition~$3$. 

The other conditions can be proved similarly. 
\end{proof}

In our prescribed representation of $T$, assuming it is a 2SUIG, bridge edges of the extended red path are induced by red branch vertices. Here we show that if two adjacent red vertices have degree 4 each, then they must be branch vertices. 

\begin{lemma}\label{lem red deg4bridge}
If two adjacent red vertices both have degree 4, then they must be in different stabs. 
\end{lemma}

\begin{proof}
Without loss of generality assume that $a_la_{l+1}$ is such an edge where $a_{l},a_{l+1}$ are 
both degree 4 lower vertices with $s_{a_l} <_x s_{a_{l+1}} $. Then either the upper-right corner of $s_{a_l}$ 
is contained in  $s_{a_{l+1}}$ or the upper-left corner of 
$s_{a_{l+1}}$ is contained in $s_{a_l}$. If the upper-right corner of $s_{a_l}$ is 
contained in  $s_{a_{l+1}}$, then $a_l$ cannot have more than three neighbors as no 
neighbor other than $s_{a_{l+1}}$ can contain one of the right corners of $s_{a_{l}}$. We can argue similarly for the other case as well. 
\end{proof}

\subsection{The canonical representation}
In this section suppose that $T$ is a tree with maximum degree 4 such that either there is no red edge or the 
red edges induces a path.
We will try to obtain a 2SUIG representation of $T$ and if our process fails to obtain such a presentation, then we will conclude that $T$ is not a 2SUIG.  Also assume that the extended red path of $T$ is $A = a_1a_2...a_k$. 
Due to Lemma~\ref{lem red monotone} and~\ref{lem red stretched} we can assume that $A$ is a stretched right monotone path  and $a_1$ is a lower vertex. 

Our strategy is to first represent $a_1$ and its associates and then to represent $a_i$ and its associates one by one in 
accending order of indices where $i \in \{2,3,...,k\}$.
In each step our strategy is to  represent $a_i$ and its associates
 in such a way that the maximum value of $x_v$  is minimized where $v$ is a vertex from  $a_i$ and its associates. 

Note that the main difficulty is to represent $a_i$ and its associates when $d(a_i) \geq 3$ as otherwise $a_i$ do not have any agents or tails. We start off with the following useful lemma.

\begin{lemma}\label{lem a1 same a2}
There exists a 2SUIG representation, satisfying all the properties of the canonical representation proved till now, with $a_1$ and $a_{2}$ in the same stab if and only if 
either $d(a_1) \leq 3$ or $d(a_2) \leq 3$. 
\end{lemma}

\begin{proof}
The ``only if'' part follows from Lemma~\ref{lem red deg4bridge}. The ``if'' part can be proved similarly like the proof of 
Lemma~\ref{lem red bridge-branch}. 
\end{proof}

The above result completely determines when $a_2$ will be in the lower stab and when it will be in the upper stab.

\subsubsection{Representation of $a_1$ and its associates when $k \neq 1$}
First we will handle the case $k \neq 1$.    Now we are going to list out the way to obtain the canonical representation  of $a_1$ and its associates and the conditions  for it to be valid through case analysis. 
 Also in any representation the agents intersecting the lower-left corner, the upper-left corner, the upper-right corner and 
the lower-right corner   of $s_{a_1}$ are renamed as  $z_1, z_2, z_3$ and $z_4$, respectively. The conditions below are simple conditions for avoiding cycles in the graph.

 \begin{itemize}
 \item[Case 1:] $d(a_1) = 4, d(a_2) = 4$.   In this case $a_2$ is an upper vertex  by Lemma~\ref{lem red deg4bridge}, $s_{a_2}$ intersects the upper-right corner of $s_{a_1}$  
 and the three agents  of $a_1$ are $z_1, z_2, z_4$. 
\begin{itemize}
\item[(1)]  $lt(z_1)$ is shrinked
lower-left monotone and $st(z_1)$ is shrinked
upper-left monotone.  

\item[(2)]  $lt(z_2)$ is shrinked
upper-left monotone and $st(z_2)$ is shrinked
upper-right monotone.

\item[(3)]  $|st(z_2)| \leq 1$ and if  $|st(z_1)| > 0$, then $|lt(z_2)| \leq 1$.

\item[(4)] $|lt(z_4)| = 0$. 
\end{itemize}

\item[Case 2:] $d(a_1) = 3, d(a_2) = 4$.   In this case $a_2$ is a lower vertex  by Lemma~\ref{lem a1 same a2}, $s_{a_2}$ intersects the upper-right corner of $s_{a_1}$  
 and the two agents of $a_1$ are $z_1, z_2$. 
\begin{itemize}
\item[(1)]  conditions (1)--(3) of Case~1.
\end{itemize}

\item[Case 3:] $d(a_1) = 4, d(a_2) = 3$.   In this case $a_2$ is a lower vertex  by Lemma~\ref{lem a1 same a2}, $s_{a_2}$ intersects the lower-right corner of $s_{a_1}$  
 and the three agents of $a_1$ are $z_1, z_2, z_3$. 
\begin{itemize}
\item[(1)]  condition (1)--(3) from Case~1.

\item[(2)]  $|lt(z_3)| \leq 1$, $lt(z_3)$ is shrinked
upper-left monotone and $st(z_3)$ is shrinked
upper-right monotone. 
\end{itemize}

\item[Case 4:] $d(a_1) = 3, d(a_2) = 3$.   In this case $a_2$ is a lower vertex  by Lemma~\ref{lem a1 same a2}, $s_{a_2}$ intersects the lower-right corner of $s_{a_1}$  
 and the two agents of $a_1$ are $z_1, z_2$. 
\begin{itemize}
\item[(1)]  $lt(z_1)$ is shrinked
lower-left monotone and $st(z_1)$ is shrinked
upper-left monotone.  

\item[(2)]  if $|st(z_1)| > 0$, then $lt(z_2)$ is shrinked
upper-right monotone with $|lt(z_2)| \leq 3$ and $st(z_2)$ is shrinked
upper-left monotone $|st(z_2)| \leq 1$.

\item[(3)]  if $|st(z_1)| = 0$, then $lt(z_2)$ is shrinked
upper-left monotone and $st(z_2)$ is shrinked
upper-right monotone $|st(z_2)| \leq 3$. 
\end{itemize}

\item[Case 5:] $d(a_1) = 4, d(a_2) = 2$.   In this case $a_2$ is a lower vertex  by Lemma~\ref{lem a1 same a2}, $s_{a_2}$ intersects the lower-right corner of $s_{a_1}$  
 and the three agents of $a_1$ are $z_1, z_2, z_3$. 
\begin{itemize}
\item[(1)]  condition (1)--(3) from Case~1.

\item[(2)]  if $|lt(z_3)| \leq 1$, then  $lt(z_3)$ is shrinked
upper-left monotone and $st(z_3)$ is shrinked
upper-right monotone.

\item[(3)]  if $|lt(z_3)| > 1$, then $|st(z_3)| \leq 1$ and $lt(z_3)$ is shrinked
upper-right monotone and $st(z_3)$ is shrinked upper-left monotone.
\end{itemize}
 
\item[Case 6:] $d(a_1) = 3, d(a_2) = 2$.   In this case $a_2$ is a lower vertex  by Lemma~\ref{lem a1 same a2}, $s_{a_2}$ intersects the lower-right corner of $s_{a_1}$  
 and the two agents of $a_1$ are from $\{z_1, z_2, z_3\}$. 
\begin{itemize}
\item[(1)] $lt(z_1)$ is shrinked
lower-left monotone and $st(z_1)$ is shrinked
upper-left monotone.  

\item[(2)]  if both $z_1,z_2$ exists and $|st(z_1)| = 0$,  then 
$lt(z_2)$ is shrinked
upper-left monotone and $st(z_2)$ is shrinked
upper-right monotone. 

\item[(3)] if both $z_1,z_2$ exists,  $|st(z_1)| > 0$ and $|lt(z_2)| \leq 1$,  then 
$lt(z_2)$ is shrinked
upper-left monotone and $st(z_2)$ is shrinked
upper-right monotone.

\item[(4)] if both $z_1,z_2$ exists,  $|st(z_1)| > 0$ and $|lt(z_2)| > 1$,  then 
$st(z_2)$ is shrinked
upper-left monotone and $lt(z_2)$ is shrinked
upper-right monotone with $|st(z_1)| \leq 1$.

\item[(5)] there is no case when both $z_2,z_3$ exists as we can always modify this representation by making the agent playing the role of $z_2$ play the role of $z_1$ instead. 

\item[(6)] there is no case when both $z_1,z_3$ exists with $|st(z_1)| = 0$ as we can always modify this representation by making the agent playing the role of $z_3$ play the role of $z_2$ instead.

\item[(7)] there is no case when both $z_1,z_3$ exists with $|st(z_3)| \leq 1$ and $|lt(z_3)| \leq 3$ as we can always modify this representation by making the agent playing the role of $z_3$ play the role of $z_2$ instead.

\item[(8)]  if both $z_1,z_3$ exists,   $|st(z_1)| > 0$ and $|lt(z_3)| \leq 3$,  then 
$lt(z_3)$ is shrinked
upper-left monotone
and $st(z_2)$ is shrinked
upper-right monotone.

\item[(9)]  if both $z_1,z_3$ exists,   $|st(z_1)| > 0$ and $|lt(z_3)| > 3$,  then 
$st(z_3)$ is shrinked
upper-left monotone with $|st(z_1)| \leq 3$
and $lt(z_2)$ is shrinked
upper-right monotone.
\end{itemize}
\end{itemize}

The square $s_{a_2}$ must intersect one of the right corners of $s_{a_1}$. In each of the cases listed above, there 
can be at most $3! = 6$ possible ways of in which the  agents of $a_1$ can play the role of $z_1,z_2,z_3, z_4$. 
 Among all possible ways those which satisfies the above conditions,  we choose the one for which the 
 leaf of the right-monotone 
 tail of  $z_2$ (only when $z_3, z_4$ does not exist) or $z_3$ ($z_4$ cannot exist) or $z_4$ ($z_3$ cannot exist) is minimized with respect to $<_x$.   
  As there are at most a constant number of probes to be made, this is achieveable in constant time. 
  Moreover,  such a representation, if found, will be called the 
  \textit{optimized representation of $a_1$ and its associates}.  Otherwise, $T$ is not a 2SUIG. 

\subsubsection{Representation of $a_1$ and its associates when $k=1$} 
Now we will handle the case $k=1$. Let $a_1$ be a lower vertex and the agents intersecting the lower-left corner, the upper-left corner, the upper-right corner and 
the lower-right corner   of $s_{a_1}$ are renamed as  $z_1, z_2, z_3$ and $z_4$, respectively.
The agents should follow the conditions listed below. These are simple conditions for avoiding cycles.

\begin{itemize}
\item[(1)]  $P_{11}$ is shrinked
lower-left monotone and $P_{12}$ is shrinked
upper-left monotone where $\{P_{11}, P_{12} \} = \{st(z_1), lt(z_1)\}$.  

\item[(2)]  $P_{21}$ is shrinked
upper-left monotone and $P_{22}$ is shrinked
upper-right monotone   where $\{P_{21}, P_{22} \} = \{st(z_2), lt(z_2)\}$.   

\item[(3)]  $P_{31}$ is shrinked
upper-left monotone and $P_{32}$ is shrinked
upper-right monotone   where $\{P_{31}, P_{32} \} = \{st(z_3), lt(z_3)\}$.    

\item[(4)]  $P_{41}$ is shrinked
lower-right monotone and $P_{42}$ is shrinked
upper-right monotone  where $\{P_{41}, P_{42} \} = \{st(z_4), lt(z_4)\}$.   

\item[(5)]  either $|P_{12}| = 0$ or  $|P_{21}| \leq 1$.

\item[(6)] if $z_2$ does not exist, then  either $|P_{12}| = 0$ or  $|P_{31}| \leq 3$.

\item[(7)] if both $z_2, z_3$ exist, then   $|P_{22}|,|P_{31}| \leq 1$.

\item[(8)]  either $|P_{42}| = 0$ or  $|P_{32}| \leq 1$.

\item[(9)] if $z_3$ does not exist, then  either $|P_{42}| = 0$ or  $|P_{22}| \leq 3$.
\end{itemize}

 As there are at most a constant number of probes to be made, this is achieveable in constant time. 
For this special case too,  such a representation, if found, will be called the 
  \textit{optimized representation of $a_1$ and its associates}.  Otherwise, $T$ is not a 2SUIG.

\subsubsection{Representation of $a_i$ and its associates for all $1 < i < k$}
  Now we will inductively describe the  canonical representation  of $a_i$ and its associates given the canonical representation of $a_j$ and its associates for all $j < i$ and the conditions  for it to be valid through case analysis. 
Note that as the way of having the canonical representation of $a_1$ and its associates is known, 
the conditions listed below is readily applicable for finding the canonical representation of $a_2$ and its associates. 
Furthermore, it is applicable for finding the canonical representation of $a_i$ and its associates by induction for all $i \in \{2,3,..., k-1\}$ as $d(a_{i+1})$ is a parameter that we need to know for finding a representation.

 Throughout the case analysis we will assume without loss of generality that $a_{i-1}$ is an upper vertex.
  Also assume that $a_{i'}$ be the maximum $i' < i$ such that  $d(a_{i'}) \geq 3$ 
 ($i-1 = i'$ is possible). 
Moreover,  in any representation the agents intersecting the lower-left corner,  the upper-right corner and 
the lower-right corner   of $s_{a_i}$ are renamed as  $z_1, z_3$ and $z_4$, respectively. The conditions below are simple conditions for avoiding cycles in the graph.

 \begin{itemize}
\item[(1)] If $z_1$ exists and $a_i$ is a lower vertex, 
then $|st(z_1)| = 0 $ 
and $lt(z_1)$ is a lower-left shrinked monotone path satisfying conditions of Lemma~\ref{lem 4options}.

\item[(2)] If $z_1$ exists, $a_i$ is an upper vertex and either $z_3$ or $z_4$ exists, 
then $st(z_1)$ is a lower-right shrinked monotone path with $|st(z_1)| \leq 1 $ 
and $lt(z_1)$ is a lower-left shrinked monotone path satisfying conditions of Lemma~\ref{lem 4options}.

\item[(3)] If $z_1$ exists, $a_i$ is an upper vertex and neither $z_3$ nor $z_4$ exists, 
$P_{11}$ is shrinked
lower-left monotone satisfying conditions of Lemma~\ref{lem 4options} and $P_{12}$ is shrinked
lower-right monotone where $\{P_{11}, P_{12} \} = \{st(z_1), lt(z_1)\}$.

\item[(4)] If $z_3$ exists with  $a_{i}$ being an  upper vertex, then  
$|st(z_3)| = 0 $ 
and $lt(z_3)$ is an upper-right shrinked monotone path.

\item[(5)] If $z_3$ exists with  $a_{i}$ being a lower vertex, then  
$z_3$ has an upper-left monotone tail $P_1$ satisfying $|P_1| \leq 1$ 
and an upper-right monotone tail $P_2$ with some 
$\{P_1, P_2\} = \{st(z_3), lt(z_3)\}$.

\item[(6)] If $z_4$ exists with  $a_{i}$ being an  upper vertex and $z_1$ also exists, then  
$z_4$ has a lower-left monotone tail $P_1$ satisfying $|P_1| \leq 1$
and a lower-right monotone tail $P_2$ with some 
$\{P_1, P_2\} = \{st(z_4), lt(z_4)\}$.

\item[(7)] If $z_4$ exists with  $a_{i}$ being an  upper vertex and $z_1$ does not exist, then  
$z_4$ has a lower-left monotone tail $P_1$ conditions of Lemma~\ref{lem 4options}
and a lower-right monotone tail $P_2$ with some 
$\{P_1, P_2\} = \{st(z_4), lt(z_4)\}$.

\item[(8)] If $z_4$ exists with  $a_{i}$ being a lower vertex, then  
$|st(z_4)| = 0 $ 
and $lt(z_4)$ is a lower-right shrinked monotone path.  
\end{itemize}

 In each of the cases listed above, there 
can be at most  
$2 \times 2 = 4$ possible ways of in which the  agents of $a_i$ can play the role of $z_1, z_3, z_4$. 
And for each agent playing the role of $z_j$ the tails can play 
the role of $P_1, P_2$ in $2$ different ways for each $j \in \{1, 3, 4\}$. 
Thus there can be at most $4 \times 2^3 = 32$ possible ways in which $a_i$ and its associates can be represented. 
 Among all possible ways those which satisfies the above conditions,  we choose the one for which the 
 leaf of the right-monotone 
 tail of  $z_3$ ($z_4$ cannot exist) or $z_4$ ($z_3$ cannot exist) is minimized with respect to $<_x$.   
  As there are at most a constant number of probes to be made, this is achieveable in constant time. 
  Moreover,  such a representation, if found, will be called the 
  \textit{optimized representation of $a_i$ and its associates}.  Otherwise, $T$ is not a 2SUIG.

\subsubsection{Representation of $a_k$ and its associates}
The canonical representation of $a_k$ and its associates is relatively simpler. Note that $a_k$ can have at most 3 agents and each of them can have at most 2 tails. 
Suppose that in any representation the agents intersecting the lower-left corner,  the upper-right corner and 
the lower-right corner   of $s_{a_k}$ are renamed as  $z_1,  z_3$ and $z_4$, respectively. 
Without loss of generality assume that $a_{k-1}$ is an upper vertex. 
Then our goal will be to find a representation of  $a_k$ and its associates which satisfies the following conditions.

 \begin{itemize}
\item[(1)] If $z_1$ exists with $a_{k}$ being a lower vertex, then 
 $|st(z_1)| = 0 $ 
and $lt(z_1)$ is a lower-left shrinked monotone path satisfying conditions of Lemma~\ref{lem 4options}.

\item[(2)] If $z_1$ exists with $a_{k}$ being an upper vertex and $z_4$ exists, then $|st(z_1)| \leq 1$.
Moreover, if $|lt(z_1)| \leq 1$, then 
$lt(z_1)$ is a lower-right shrinked monotone path  
and $st(z_1)$ is a lower-left shrinked monotone path satisfying conditions of Lemma~\ref{lem 4options}.  
Otherwise,
 $st(z_1)$ is a lower-right shrinked monotone path 
and $lt(z_1)$ is a lower-left shrinked monotone path satisfying conditions of Lemma~\ref{lem 4options}.

\item[(3)] If $z_1$ exists with $a_{k}$ being an upper vertex, $z_4$ does not exist and $z_3$ exists with 
$|st(z_3)| \geq 1$, then $|st(z_1)| \leq 3$.
Moreover, if $|lt(z_1)| \leq 3$, then 
$lt(z_1)$ is a lower-right shrinked monotone path  
and $st(z_1)$ is a lower-left shrinked monotone path satisfying conditions of Lemma~\ref{lem 4options}.  
Otherwise,
 $st(z_1)$ is a lower-right shrinked monotone path 
and $lt(z_1)$ is a lower-left shrinked monotone path satisfying conditions of Lemma~\ref{lem 4options}.

\item[(4)] If $z_1$ exists with $a_{k}$ being an upper vertex, $z_4$ does not exist and $z_3$ exists with 
$|st(z_3)| = 0$,  then
$lt(z_1)$ is a lower-right shrinked monotone path  
and $st(z_1)$ is a lower-left shrinked monotone path satisfying conditions of Lemma~\ref{lem 4options}.

\item[(5)] If $z_3$ exists with  $a_{k}$ being an  upper vertex, then  
$st(z_3)$ is a lower-right shrinked monotone path
and $lt(z_3)$ is an upper-right shrinked monotone path.
Moreover, if $z_4$ exists, then  $|st(z_3)| = 0 $.

 \item[(6)] If $a_{k-1}$ has an upper-right monotone tail $P$ with $|P| \geq 2$ and $a_{k}$ is a lower vertex, then $z_3$
 cannot  exist.

 \item[(7)]  If $z_3$ exists with  $a_{k}$ being a lower vertex and $z_4$ exists with $|lt(z_4)| \geq 1$, then  
$z_3$ has an upper-left monotone tail $P_1$ satisfying $|P_1| \leq 1$
and an upper-right monotone tail $P_2$ with some 
$\{P_1, P_2\} = \{st(z_3), lt(z_3)\}$.

 \item[(8)]  If $z_3$ exists with  $a_{k}$ being a lower vertex and 
 if either $z_4$ exists with $|lt(z_4)| = 0$ or $z_4$ does not exist, then  
$lt(z_3)$ is an upper-right monotone tail and $st(z_3)$ is  a lower-right monotone tail.

\item[(9)] If $z_4$ exists with  $a_{k}$ being an  upper vertex and $z_3$ exists with $|lt(z_3)| \geq 1$, then  
$z_4$ has a lower-left monotone tail $P_1$ satisfying conditions of Lemma~\ref{lem 4options}
and a lower-right monotone tail $P_2$ with some 
$\{P_1, P_2\} = \{st(z_4), lt(z_4)\}$.

\item[(4)] If $z_4$ exists with  $a_{k}$ being an  upper vertex and if either $z_3$ exists with $|lt(z_3)| = 0$ or $z_3$ does not exist, then  
$st(z_4)$ is an upper-right monotone tail 
and $st(z_4)$ is an upper-right monotone tail. 

\item[(5)] If $z_4$ exists with  $a_{k}$ being a lower vertex, then  
$st(z_4)$ is a lower-right shrinked monotone path. 
and $lt(z_4)$ is a lower-right shrinked monotone path.  
Moreover, if either $z_3$ exists or $z_2$ does not exist but $a_{k-1}$ has an upper-right monotone tail $P$ with $|P| \geq 4$, then $|st(z_4)| = 0$.
\end{itemize}

 In each of the cases listed above, there 
can be at most  
$3! = 6$ possible ways of in which the  agents of $a_1$ can play the role of $z_1, z_3, z_4$. 
For each agent playing the role of $z_j$ the tails can be placed in at most two 
 different ways for each $j \in \{1, 3, 4\}$. 
Thus there can be at most $6 \times 2^3 = 48$ possible ways in which $a_i$ and its associates can be represented. 
If we can find one representation  which satisfies the above conditions,  we choose that one 
and call it \textit{optimized representation of $a_k$ and its associates}. 
Otherwise, $T$ is not a 2SUIG. 
 As there are at most a constant number of probes to be made, this is achieveable in constant time. 

\subsection{Algorithm}
Finally we will  describe the algorithm for recongnizing if a given tree $T$ is a 2SUIG. 
Whenever our algorithm concludes that the given tree $T$ is not a 2SUIG, there is a configuration responsible for it. 
These configurations are forbidden configurations for 2SUIG trees.

\begin{itemize}
\item[(1)]  Check if maximum degree of $T$ is at most 4. If not, then $T$ is not a 2SUIG by Lemma~\ref{obs_deg_4}. Otherwise, go to the next step.

\item[(2)] Check if there at most one branch vertex in $T$. If yes, then $T$ is a 2SUIG by Lemma~\ref{lm_1branch}. Otherwise, go to the next step.

\item[(3)]  Find out the graph induced by the red edges of the tree. If that graph has at  least one  edge but not a path, then 
$T$ is not a 2SUIG by Lemma~\ref{lm_RUIC}. Otherwise, go to the next step.

\item[(4)]  Find out a (but for some trivial cases it is unique) extended  red path $A = a_1a_2...a_k$. 
Assign $x_{a_i} = i$ for all $i \in \{1,2, ..., k\}$. Moreover, put $s_{a_1}$ in the lower stab.  

\item[(5)]    For $i = 1 \text{ to } k$ find out the 
optimized representation of $a_i$ and its associates. If we fail to find such a representation for some $i \in \{1, 2, ..., k\}$, then $T$ is not a 2SUIG. 
\end{itemize}

Correctness of the algorithm implies from the previous results and discussions.
Given a tree it is possible to find out its set of red edges in linear time using post-order traversal. 
For the other steps
we need to probe at most a constant number cases  for each red vertex. Thus, it is possible to run the algorithm in 
$O(|V|)$ time.

\section{Conclusion}
In this paper we consider the problem of recognizing 2SUIG trees. While doing that we proved a number of 
structural properties and provided insights regarding how a canonical 2SUIG representation of a tree can be obtained. 
Recall our discussion on red edges and red vertices of a tree. Observe that, if the red vertices induce a path, 
then the tree has a unit square intersection representation.  
Hence, we hope our work can be extended for ``k-stab unit interval graphs'' and will help solving 
the tree recognition problem for cubicity two graphs. Moreover, our graph class is obtained by 
putting edges between two unit interval graphs. From applications point of view, as 
unit interval graphs have wide range of applications, our graph class might be able to capture 
interactions between two unit interval graphs  and prove to be valuable in future.

\bibliographystyle{abbrv}
\bibliography{science}

\begin{thebibliography}{1}

\bibitem{Chandran2014}
J.~Babu, M.~Basavaraju, L.~S. Chandran, D.~Rajendraprasad, and N.~Sivadasan.
\newblock Approximating the cubicity of trees.
\newblock {\em CoRR}, abs/1402.6310, 2014.

\bibitem{Bhore2015}
S.~Bhore, D.~Chakraborty, S.~Das, and S.~Sen.
\newblock On a special class of boxicity 2 graphs.
\newblock In {\em Algorithms and Discrete Applied Mathematics}, volume 8959 of
  {\em Lecture Notes in Computer Science}, pages 157--168. Springer
  International Publishing, 2015.

\bibitem{breu}
H.~Breu.
\newblock {\em Algorithmic aspects of constrained unit disk graphs}.
\newblock PhD thesis, University of British Columbia, 1996.

\bibitem{golumbic}
M.~Golumbic.
\newblock {\em Algorithmic Graph Theory and Perfect Graphs: Second Edition}.
\newblock Annals of Discrete Mathematics. Elsevier Science, 2004.

\bibitem{krato}
J.~Kratochv{\'\i}l.
\newblock A special planar satisfiability problem and a consequence of its
  np-completeness.
\newblock {\em Discrete Applied Mathematics}, 52(3):233--252, 1994.

\bibitem{roberts}
F.~S. Roberts.
\newblock On the boxicity and cubicity of a graph.
\newblock {\em Recent Progresses in Combinatorics}, pages 301--310, 1969.

\end{thebibliography}

\end{document}